\documentclass[9pt, shortpaper,twoside,web]{ieeecolor}
\usepackage{generic}

\usepackage{amsmath,amssymb,amsfonts}
\usepackage{hyperref}
\usepackage[nocompress]{cite}
\usepackage{balance}

\usepackage[caption=false,font=footnotesize]{subfig}
\usepackage{graphicx}
\usepackage{tikz, pgfplots}
\usetikzlibrary{shapes, arrows, arrows.meta, fit, positioning, calc, 3d}
\usepgfplotslibrary{fillbetween}
\pgfplotsset{compat=1.18} 
\tikzset{>=latex}

\DeclareMathOperator*{\argmin}{arg\,min}
\newtheorem{theorem}{Theorem}
\newtheorem{lemma}{Lemma}

\newtheorem{definition}{Definition}
\newtheorem{proposition}{Proposition}
    
\begin{document}
\bstctlcite{IEEEexample:BSTcontrol}

\title{Remote State Estimation over a Wearing Channel: \\Information Freshness vs. Channel Aging}
\author{Jiping~Luo, George~Stamatakis, Osvaldo~Simeone, \IEEEmembership{Fellow, IEEE}, \\and Nikolaos~Pappas, \IEEEmembership{Senior Member, IEEE}
\thanks{The work of J.~Luo and N.~Pappas was supported by ELLIIT, the Graduate School in Computer Science (CUGS), and the European Union (6G-LEADER, 101192080). \textit{(Corresponding author: Jiping~Luo.)}}
\thanks{The work of O. Simeone was supported by the European Research Council (ERC) under the European Union's Horizon Europe Programme (grant agreement No. 101198347), by an Open Fellowship of the EPSRC (EP/W024101/1), and by the EPSRC project (EP/X011852/1).}
\thanks{J.~Luo and N.~Pappas are with the Department of Computer and Information Science, Link\"oping University, Link\"oping 58183, Sweden. (e-mail: jiping.luo@liu.se; nikolaos.pappas@liu.se).}
\thanks{G.~Stamatakis is with the Institute of Computer Science, 
Foundation for Research and Technology -- Hellas (ICS-FORTH), Crete 70013, Greece (e-mail: gstam@ics.forth.gr).}
\thanks{O. Simeone is with the Institute for Intelligent Networked Systems, Northeastern University London, One Portsoken Street, London, E1 8PH, UK (e-mail: o.simeone@northeastern.edu).}
}

\maketitle
\begin{abstract}
We study the remote estimation of a linear Gaussian system over a channel that wears out over time and with every use. The sensor can either transmit a fresh measurement in the current time slot, restore the channel quality at the cost of downtime, or remain silent. Frequent transmissions yield accurate estimates but incur significant wear on the channel. Renewing the channel too often improves channel conditions but results in poor estimation quality. What is the optimal timing to transmit measurements and restore the channel? This problem is formulated as a semi-Markov decision process (SMDP). We establish monotonicity properties of the optimal policy and propose structure-aware solution methods.
\end{abstract}

\section{Introduction}\label{sec:introcution}
\subsection{Motivation and Contributions}
Channel unreliability is a salient feature and a bottleneck in networked control systems (NCSs)~\cite{HenrikCommSurvey2018}. Wireless channels are inherently less reliable than their wired counterparts~\cite{LucaProcIEEE2007}. Data loss, delay, and nonstationarity can significantly degrade system performance. A fundamental yet challenging problem is the remote state estimation of dynamic processes that is robust to the realities of wireless networks.

Remote estimation and control over lossy channels have been extensively studied in the literature. Data loss is modeled stochastically by assigning failure probabilities to different channel states. A common model is the independent and identically distributed (\emph{i.i.d.}) Bernoulli process (see, e.g.,~\cite{Luca2008TAC, shi2011Auto, Leong2017TAC, Shi2020TCNS}), where each packet is dropped independently with a fixed probability. The Markovian channel, modeled by a time-homogeneous finite-state Markov process, captures temporal correlations in packet losses~\cite{Dey2007Auto, Wanchun2022TAC}. 

In contrast, this paper studies remote estimation over \emph{wearing channels} whose quality deteriorates due to natural aging and usage. Such degradation phenomena are often observed in wireless communication systems with low-power devices or those exposed to harsh environments, where frequent communication or operational stress depletes energy reserves, accelerates hardware aging, and further impairs link reliability~\cite{wearingChannel}. Similar effects occur in neural connections~\cite{abraham1996metaplasticity} and quantum channels~\cite{schlosshauer2019quantum}. Neural efficiency declines with age, weakening synaptic strength and neural connections. Repeated use or neurological conditions can further impair synaptic structure and function. In quantum channels, the quality of entanglement and the coherence of qubit pairs degrade over time due to environmental noise and interactions. Restoring channel quality requires external intervention, such as charging or replacing devices, or generating fresh entangled qubit pairs.

Our main contributions are as follows. We formulate an optimization problem to minimize the average estimation error over a wearing channel. The channel reliability is modeled as a decreasing function of age, which grows with time and usage. It is possible to renew the channel at the cost of downtime. This problem is a semi-Markov decision process (SMDP) with a countably infinite state space that features two dependent age processes: the \emph{Age of Information} (AoI) and the \emph{Age of Channel} (AoC). We show the existence of an AoC-monotone optimal policy; that is, the optimal policy is weakly increasing (from idle to transmit to renewal) in AoC for any fixed AoI.  Building on these findings, we propose a structure-aware algorithm to compute the optimal solution with reduced computation overhead.

\subsection{Related Work}
Information \emph{freshness} is closely related to \emph{accuracy} in the remote estimation of linear Gaussian processes. In systems where the sensor runs a Kalman filter to pre-estimate the source states and transmits these estimates, rather than raw measurements, to a remote receiver, the error covariance at the receiver is a monotonically nondecreasing function of the time elapsed since the last received packet was generated~\cite{Shi2020TCNS, Wanchun2022TAC}. This draws a connection between distortion and AoI, which measures how fresh the receiver's knowledge is about an information source~\cite{yatesInfocom}. Prior studies have shown that, to achieve an optimal tradeoff between estimation performance and transmission cost, the optimal transmission policy has a threshold structure: the sensor transmits once the error covariance exceeds a threshold that depends on channel reliability~\cite{Luca2008TAC, shi2011Auto, Leong2017TAC, Shi2020TCNS, Dey2007Auto, Wanchun2022TAC}. Since the estimation error can grow unbounded due to packet losses, sufficient and necessary conditions for mean-square stability have been established for \emph{i.i.d.}~\cite{Luca2008TAC} and Markovian~\cite{Dey2007Auto, Wanchun2022TAC} channel models.

This work can be considered as a contribution to the rich literature on AoI. AoI minimization has been extensively studied in prior work, including~\cite{pappas2023age, YinSun2021TIT, Yates2021JSAC, Modiano2024TON}. Beyond information freshness, recent studies have introduced semantics-aware metrics to capture the contextual relevance and goal-oriented usefulness of information in the remote estimation of Markov processes~\cite{luo2024semantic, luo2024exploiting, luo2024cost}. The closest study to this work is~\cite{George2023Asilomar}, which investigates AoI minimization over wearing channels. However, remote estimation and control over wearing channels remain largely unexplored.

\begin{figure}[t!]
    \centering
    \includegraphics[width=\linewidth]{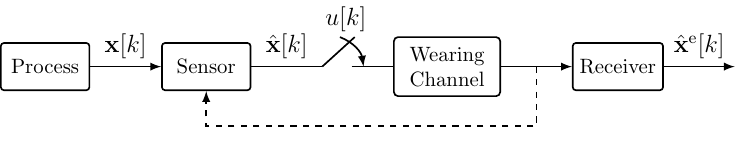}
    \caption{The remote estimation system with a wearing channel.}
    \label{fig:system_model}
\end{figure}

\section{System Model and Problem Formulation}\label{sec:system_model}
\subsection{System Setup}\label{sec:system_setup}
We consider the remote estimation system depicted in Fig.~\ref{fig:system_model}.
\subsubsection{Process}
The source is modeled as a linear Gaussian process
\begin{equation}
    \mathbf{x}[k+1] = \mathbf{A} \mathbf{x}[k] + \mathbf{w}[k], \quad k = 0, 1, 2, \ldots \label{eq:process}
\end{equation}
where $\mathbf{x}[k] \in \mathbb{R}^{l}$ is the process state at time slot $k$, $\mathbf{A} \in\mathbb{R}^{l \times l}$ is the state transition matrix, and $\mathbf{w}[k] \in \mathbb{R}^{l}$ is the process noise. The sequence $\{\mathbf{w}[k]\}$ is assumed to be \emph{i.i.d.} Gaussian with zero mean and covariance matrix $\mathbf{Q}\succeq 0$.

\subsubsection{Sensor}
The process in~\eqref{eq:process} is measured by a sensor as
\begin{equation}
    \mathbf{y}[k] = \mathbf{C}\mathbf{x}[k] + \mathbf{v}[k], 
\end{equation}
where $\mathbf{y}[k] \in\mathbb{R}^{m}$ is the noisy measurement at time slot $k$, $\mathbf{C} \in \mathbb{R}^{m \times l}$ is the measurement matrix, and $\mathbf{v}[k] \in\mathbb{R}^{m}$ is the measurement noise. The sequence $\{\mathbf{v}[k]\}$ is assumed to be \emph{i.i.d.} Gaussian with zero mean and covariance matrix $\mathbf{R} \succ 0$. Moreover, the noise processes $\{\mathbf{w}[k]\}$ and $\{\mathbf{v}[k]\}$ are assumed to be mutually independent. 

After taking a measurement $\mathbf{y}[k]$, the sensor runs a Kalman filter to compute the minimum mean-square error (MSE) estimate of $\mathbf{x}[k]$ using measurements up to and including time slot $k$. The output of the Kalman filter is the state estimate $\hat{\mathbf{x}}[k]$ with the estimation error covariance $\mathbf{P}[k]$, where\cite{shi2011Auto}
\begin{align}
    \hat{\mathbf{x}}[k]
    &=\hat{\mathbf{x}}^{-}[k] + \mathbf{G}[k] (\mathbf{y}[k] -\mathbf{C} \hat{\mathbf{x}}^{-}[k]),\\
    \mathbf{P}[k]
    &=(\mathbf{I} - \mathbf{G}[k] \mathbf{C}) \mathbf{P}^{-}[k].
\end{align}
Here, $\hat{\mathbf{x}}^{-}[k]$, $\mathbf{P}^{-}[k]$, and $\mathbf{G}[k]$ are the \emph{a priori} state estimate, the \emph{a priori} estimation error covariance, and the optimal filter gain at time slot $k$, respectively, given by
\begin{subequations}
\begin{align}
    \hat{\mathbf{x}}^{-}[k]
    &=\mathbf{A} \hat{\mathbf{x}}[k-1],\\
    \mathbf{P}^{-}[k]
    &=\mathbf{A} \mathbf{P}[k-1] \mathbf{A}^{\top} + \mathbf{Q},\\
    \mathbf{G}[k] 
    &=\mathbf{P}^{-}[k] \mathbf{C}^{\top}(\mathbf{C} \mathbf{P}^{-}[k] \mathbf{C}^{\top}
    + \mathbf{R})^{-1}.
\end{align}
\end{subequations}

Under the standard assumptions that the pair $(\mathbf{A}, \mathbf{C})$ is observable and $(\mathbf{A},\sqrt{\mathbf{Q}})$ is controllable, the error covariance $\mathbf{P}[k]$ converges exponentially fast to the steady-state value $\bar{\mathbf{P}}$ as $k \to \infty$~\cite[Ch.~4.4]{kalman_filter}. In the rest of the paper, we assume that the local Kalman filter operates in steady state, i.e., $\mathbf{P}[k]=\bar{\mathbf{P}}$ for all $k$.

\subsubsection{Wearing channel} 
The sensor decides whether to transmit its estimate $\hat{\mathbf{x}}[k]$ to the receiver through a channel that \emph{wears out over time and with every use}. Let $u[k] \in \mathcal{U}$ denote the sensor’s decision at time slot $k$, where $\mathcal{U} = \{0, 1, 2\}$, and
\begin{itemize}
    \item $u[k] = 0$ denotes the \emph{idle} action. The channel is not used, and its reliability degrades due to natural aging.
    \item $u[k] = 1$ denotes the \emph{transmit} action. Each transmission incurs a certain amount of wear on the channel, leading to an additional reduction in its reliability.
    \item $u[k] = 2$ denotes the \emph{renewal} action. Maintenance or replacement is performed to restore the channel quality; however, this action takes time to complete.
\end{itemize}

We assume that data transmission takes one slot to complete, whereas channel renewal occupies $\delta_{\textrm{R}}>1$ consecutive slots. During renewal, the sensor must remain idle; that is, a new decision is made only after the current action is completed. For ease of analysis, we introduce the decision epoch $t=0,1,\ldots$. Throughout this paper, square brackets, e.g., $u[k]$, index quantities evolving at every time slot $k=0,1,\ldots$, whereas subscripts, e.g., $u_{t}$, index quantities evaluated at decision epochs $t=0,1,\ldots$. 

The sojourn time between the $t$th and $(t+1)$th decision epochs is
\begin{equation}
    \ell(u_{t}) = \begin{cases}
        1, &u_{t} \in\{0, 1\},\\
        \delta_{\textrm{R}}, &u_{t} = 2.
    \end{cases}
\end{equation}
Then the AoC is recursively defined at decision epochs as
\begin{equation}
    \tau_{t+1} := \begin{cases}
        \tau_t +1, &u_t = 0,\\
        \tau_t +\tau_\text{D}, &u_t = 1,\\
        1, &u_t = 2,\\
    \end{cases}
\end{equation}
where $\tau_\text{D}>1$ represents the amount of wear incurred by each transmission. The AoC summarizes the accumulated channel degradation caused by natural aging and channel usage. Fig.~\ref{fig:AoC} shows the evolution of the AoC over different timelines.

\begin{figure*}[t!]
\centering
\subfloat[The evolution of AoC at the transmitter.]{
\includegraphics[width=0.45\linewidth]{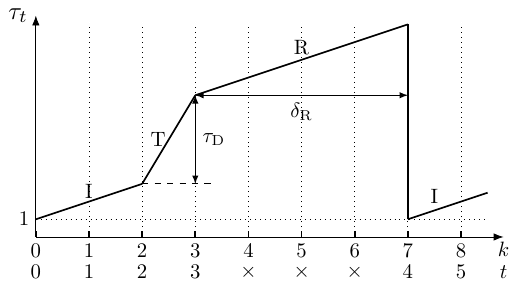}
\label{fig:AoC}
}
\hspace{2em}
\subfloat[The evolution of AoI at the receiver.]{
\includegraphics[width=0.45\linewidth]{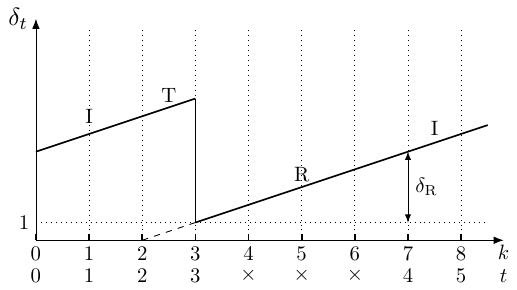}
\label{fig:AoI}
}
\caption{An illustration of the timelines and the evolution of the age processes, where `I', `T', and `R' denote idle, transmit, and renewal actions, respectively. The sensor transmits in the $2$nd slot and renews the channel at the $3$rd slot. Since data transmission takes one slot, the AoI resets to one at the $3$rd decision epoch. During renewal, the sensor remains silent, and the AoI continues to increase.}
\label{fig:age-processes}
\end{figure*}

Let $h_{t}=1$ denote a successful transmission at decision epoch $t$, and $h_{t}=0$ denote an unsuccessful transmission. The channel reliability is time-varying and deteriorates with the AoC. In particular, the probability of a successful transmission at the $t$th epoch is
\begin{equation}
    \Pr[h_{t} = 1| u_{t} = 1] = \theta(\tau_{t}),
\end{equation}
where $\theta: \mathbb{N}^{+} \to [\theta_{\min}, \theta_{\max}]$ is a monotonic nonincreasing and bounded function of the AoC, and $\theta_{\min}$ and $\theta_{\max}\in[0,1]$ represent the worst and best channel conditions, respectively. The probability of an unsuccessful transmission at the $t$th epoch is 
\begin{equation}
    \Pr[h_{t} = 0| u_{t} = 1] = 1 - \theta(\tau_{t}) = \bar{\theta}(\tau_{t}).
\end{equation}

\subsubsection{Receiver}
At each time $k$, the receiver either reconstructs the source state using the newly received packet $\hat{\mathbf{x}}[k]$ or predicts the state based on outdated information. Let $\eta[k] \in \{0,1\}$ indicate whether a transmission is successful at time $k$, where $\eta[k]=1$ if $u[k]=1$ and $h[k]=1$, and $\eta[k]=0$ otherwise. The receiver's estimate evolves as
\begin{equation}
    \hat{\mathbf{x}}^{\textrm{e}}[k+1] = \eta[k]
    \mathbf{A}\hat{\mathbf{x}}[k] + 
	(1 -\eta[k])\mathbf{A} \hat{\mathbf{x}}^{\textrm{e}}[k].
\end{equation}
Provided that the Kalman filter is in steady state, the error covariance at the receiver is given by
\begin{equation}
    \mathbf{P}^{\textrm{e}}[k+1] = \begin{cases}
        \mathbf{A}\bar{\mathbf{P}} \mathbf{A}^{\top} + \mathbf{Q}, &\eta[k] = 1,\\
        \mathbf{A}\mathbf{P}^{\textrm{e}}[k]\mathbf{A}^{\top} + \mathbf{Q}, &\eta[k] = 0.
    \end{cases}\label{eq:error-covariance}
\end{equation}

Define the AoI at the receiver as the time elapsed since the last reception of a packet, i.e.,
\begin{equation}
    \delta[k] := k - \max\{\tau \leq k: \eta[\tau] = 1\}.
\end{equation}
As depicted in Fig.~\ref{fig:AoI}, the AoI increases linearly with time and resets to one only when a new packet is received. The AoI at each decision epoch $t$ evolves as
\begin{equation}
    \delta_{t+1} = \begin{cases}
        1, &\eta_{t} = 1,\\
        \delta_{t} + \delta_{\textrm{R}}, &u_{t} = 2,\\
        \delta_{t} + 1, &\textrm{otherwise}.
    \end{cases}
\end{equation}
The receiver-side error covariance at each epoch $t$ is obtained by applying the recursion in~\eqref{eq:error-covariance} for $\delta_{t}$ consecutive slots starting from the steady-state local covariance $\bar{\mathbf{P}}$. By unrolling this recursion, the error covariance can be written in terms of the AoI as
\begin{equation}
    \mathbf{P}^{\textrm{e}}_{t} =
	\mathbf{A}^{\delta_{t}} \bar{\mathbf{P}} (\mathbf{A}^{\top})^{\delta_{t}} + \sum_{r = 0}^{\delta_{t}-1}\mathbf{A}^{r} \mathbf{Q} (\mathbf{A}^{\top})^{r}.
\end{equation}
\begin{lemma}[\hspace{-0.03em}\cite{Leong2017TAC}]\label{lemma:monotonicity-of-AoI}
The estimation MSE at the receiver, i.e.,
\begin{equation}
    \operatorname{Tr}(\mathbf{P}^{\textrm{e}}_{t}) =: f(\delta_{t}),
\end{equation}
is monotonically nondecreasing in the AoI.
\end{lemma}

\textit{Remark:} Lemma~\ref{lemma:monotonicity-of-AoI} implies that sensor measurements are more valuable when they are fresh. This draws a connection between distortion and information age~\cite{luo2024exploiting, luo2024cost}. Although this paper focuses primarily on remote estimation, the proposed framework and results naturally extend to the broader AoI literature.

The receiver sends an acknowledgment (ACK) to the sensor upon the successful reception of an update packet. ACK packets are assumed to be instantaneous and error-free. Hence, the sensor knows the remote estimate at the receiver. Let $s_{t} = (\tau_{t}, \delta_{t})$ denote the system state at the $t$th epoch, where $s_{t}\in\mathcal{S}$ and the state space is $\mathcal{S} = \mathbb{N}^{+}\times \mathbb{N}^{+}$. Since actions may span different numbers of slots, we define the one-stage cost of taking action $u_{t}$ in state $s_{t}$ as the cumulative cost incurred before decision epoch $t+1$, i.e.,
\begin{equation}
\tilde{c}(s_{t}, u_{t}) = \begin{cases}
    f(\delta_{t}), & u_{t} = 0,\\
    f(\delta_{t}) + E_{\textrm{T}}, &u_{t} = 1,\\
    \sum_{r=0}^{\delta_{\textrm{R}}-1} f(\delta_{t} + r) + E_{\textrm{R}}, &u_{t} = 2,
\end{cases}\label{eq:cost-function}
\end{equation}
where $E_{\textrm{T}}$ and $E_{\textrm{R}}$ are the resource utilization costs associated with the transmit and renewal actions, respectively.

\subsection{Problem Formulation}
We aim to optimize the estimation quality while accounting for channel aging. The sensor's memory at decision epoch $t$ is
\begin{equation*}
    i_{t} = (\tau_{0:t}, \delta_{0:t}, u_{0:t-1}).
\end{equation*}
Using this information, the sensor selects an action $u_{t}$ according to a \emph{decision rule} $\pi_{t}$, i.e.,
\begin{equation}
    u_{t} = \pi_{t}(i_{t}) = \pi_{t}(\tau_{0:t}, \delta_{0:t}, u_{0:t-1}).
\end{equation}
A \emph{policy} $\pi = (\pi_{0}, \pi_{1}, \ldots)$ is a sequence of decision rules. We call a policy \emph{Markovian} if, for every $t\geq 0$, it selects an action based only on the current system state, that is, $u_{t} = \pi_{t}(s_{t})$. A Markovian policy is \emph{deterministic} if, for every $s_t\in\mathcal{S}$, it selects an action with probability one; otherwise, it is \emph{randomized}. Let $\Pi$ and $\Pi^{\textrm{MD}}$ denote the set of all admissible and Markovian deterministic policies, respectively.

The average estimation MSE of a policy $\pi\in\Pi$ over an infinite horizon is defined as
\begin{equation}
    J^{\pi}(s_{0}) := \limsup_{T\to\infty}\frac{\mathbb{E}^{\pi} \big[\sum_{t=0}^{T-1}\tilde{c}(s_{t}, u_{t}) {\,}| {\,} s_{0}\big]}{\mathbb{E}^\pi \big[\sum_{t=0}^{T-1}\ell( u_{t}){\,}| {\,} s_{0}\big]}. \label{eq:total-cost}
\end{equation}
The goal is to find a policy $\pi^{*}$ that minimizes~\eqref{eq:total-cost}, i.e.,
\begin{equation}
    J^{*}(s_{0}) = \inf_{\pi\in\Pi} J^{\pi}(s_{0}).\label{problem:MDP}
\end{equation}

Problem~\eqref{problem:MDP} is a semi-Markov decision process (SMDP) characterized by $(\mathcal{S}, \mathcal{U}, \tilde{P}, \tilde{c}, \ell)$, where $\mathcal{S}$ and $\mathcal{U}$ are the state space and action space, $\tilde{c}(s, u)$ is the cost function given by~\eqref{eq:cost-function}, $\ell(u)$ is the sojourn time, and $\tilde{P}:\mathcal{S} \times \mathcal{U}\times\mathcal{S} \to [0,1]$ is the state transition probability function. Based on the model described in Sec.~\ref{sec:system_setup}, $\tilde{P}$ is obtained as
\begin{align}
    \tilde{P}_{s,s^{\prime}}(u) &:= \Pr[s_{t+1}=s^{\prime}|s_{t} = (\tau, \delta), u_{t} = u] \notag\\
    &=\begin{cases}
        1, &u=0, s^{\prime}=(\tau+1, \delta+1),\\
        \theta(\tau), &u=1, s^{\prime}=(\tau+\tau_{\textrm{D}}, 1),\\
        \bar{\theta}(\tau), &u=1, s^{\prime} =(\tau + \tau_{\textrm{D}}, \delta + 1),\\
        1, &u=2, s^{\prime} = (1, \delta + \delta_{\textrm{R}}),\\
        0, &\textrm{otherwise}.
    \end{cases} \label{eq:dynamics}
\end{align}

\textit{Discussions on the SMDP:}

(1) There are two types of costs associated with data transmission: energy consumption and channel wear. The channel quality, and hence the state transition probabilities in~\eqref{eq:dynamics}, depend on the history of past actions since the most recent channel renewal. This dependence makes our problem significantly more challenging than existing studies that assume \emph{i.i.d.} or Markovian channel models.

(2) When the spectral radius of the state transition matrix $\mathbf{A}$ satisfies $\rho(\mathbf{A}) \geq 1$, the cost function $\tilde{c}(s,u)$ may become unbounded and grow exponentially fast with the AoI~\cite{Luca2008TAC}. Consequently, the system may become unstable (see Definition~\ref{def:stability}). Moreover, Problem~\eqref{problem:MDP} faces both computational and memory challenges because the state space is countably infinite. Therefore, a theoretical analysis of the existence and structure of the optimal policy is essential.

\begin{definition}[Mean-square stability]\label{def:stability}
The remote estimation system described in Sec.~\ref{sec:system_setup} is mean-square stable under a policy $\pi$ if the MSE is finite, i.e., $J^{\pi}(s_{0}) < \infty$ for all initial state $s_{0} \in \mathcal{S}$.
\end{definition}

\section{Main Results}\label{sec:main results}
This section presents structural results and develops structure-aware algorithms for the SMDP~\eqref{problem:MDP}.

\subsection{Existence of an Optimal Policy}
When $\rho(\mathbf{A})<1$, the cost function $\tilde{c}(s,u)$ is bounded, and hence the system is always stable. Therefore, in what follows, we focus on the more challenging case where $\rho(\mathbf{A}) \geq 1$. 

The following lemma states that the system can be stabilized even without using the renewal action if it is stabilizable under the worst channel conditions. This result stems from prior findings~\cite{Luca2008TAC} that a sufficient condition for mean-square stability over an \emph{i.i.d.} channel with constant reliability $\Theta$ is $\rho^{2}(\mathbf{A})(1 - \Theta) < 1$.

\begin{lemma}\label{lemma:without-renewal}
If $\rho^{2}(\mathbf{A})(1 - \theta_{\min}) < 1$, the system can be stabilized without using the renewal action. 
\end{lemma}

Lemma~\ref{lemma:without-renewal} implies that mean-square stability is achieved if the system state is confined within the region
\begin{equation*}
    \underline{\mathcal{S}} = 
    \{s \in \mathcal{S}: \rho^{2}(\mathbf{A})(1 - \theta(\tau))< 1 \} \neq \emptyset.
\end{equation*}
Let $\overline{\mathcal{S}} = \mathcal{S} \setminus \underline{\mathcal{S}}$ denote the unstable region. Consider the policy
\begin{equation*}
    \pi^{\prime}(s) = \begin{cases}
        1~(\text{transmit}), &s \in \underline{\mathcal{S}},\\
        2~(\text{renew}), &s \in \overline{\mathcal{S}},
    \end{cases}
\end{equation*}
which transmits when the state is in the stabilizable region $\underline{\mathcal{S}}$ and renews the channel whenever the state enters the unstable region $\overline{\mathcal{S}}$.

Each time the system reaches $\overline{\mathcal{S}}$, the renewal action resets the AoC to one, thereby restoring the channel reliability to $\theta_{\max}$. If $\rho^{2}(\mathbf{A})(1-\theta_{\max})<1$, then the post-renewal state belongs to the stabilizable region. Since renewal is completed in a finite number of slots, the system repeatedly returns to a stabilizable region in finite time. Consequently, the expected error covariance remains bounded under policy $\pi^{\prime}$. This implies the following result.

\begin{lemma}\label{lemma:with-renewal}
If $\rho^{2}(\mathbf{A})(1 - \theta_{\max}) < 1$, the system can be stabilized by leveraging renewal actions.
\end{lemma}

As a consequence of Lemma~\ref{lemma:with-renewal}, we establish the existence of a deterministic optimal policy to problem~\eqref{problem:MDP} as follows.

\begin{theorem}\label{theorem:existence}
If $\rho^{2}(\mathbf{A})(1 - \theta_{\max}) < 1$, there exists a constant $\lambda^{*}$ and a value function $V$ that satisfies Bellman's equation
\begin{equation}
    V(s) = \min_{u \in \mathcal{U}} \left[\tilde{c}(s,u) - \lambda^* \ell(u) + \sum\nolimits_{s^{\prime}}
    \tilde{P}_{s,s^{\prime}}(u)V(s^{\prime})\right], \label{eq:bellman-equation-prime}
\end{equation}
where $\lambda^{*}$ is the minimal average cost. Define
\begin{equation}
    c(s, u) = \frac{\tilde{c}(s,u)}{\ell(u)}, \,\,\, P_{s,s^{\prime}}(u) = \begin{cases}
        \frac{\tilde{P}_{s,s^{\prime}}(u)}{\ell(u)}, &s \neq s^{\prime},\\
        1-\frac{1 - \tilde{P}_{s,s}(u)}{\ell(u)}, &s= s^{\prime}.
    \end{cases}
\end{equation}
Then the SMDP $(\mathcal{S}, \mathcal{U}, \tilde{P}, \tilde{c}, \ell)$ is equivalent to the MDP $(\mathcal{S}, \mathcal{U}, P, c)$. The optimality equation~\eqref{eq:bellman-equation-prime} can be simplified to
\begin{equation}
    \lambda^{*} + V(s) = \min_{u\in \mathcal{U}}\left[c(s,u) + \sum\nolimits_{s^{\prime}}
    P_{s,s^{\prime}}(u) V(s^{\prime}) \right].\label{eq:bellman-equation}
\end{equation}
Moreover, any policy $\pi^{*} \in\Pi^{\textrm{MD}}$ that minimizes~\eqref{eq:bellman-equation} is optimal.
\end{theorem}
\begin{proof}
See Appendix~\ref{proof:theorem-existence}.
\end{proof}

Theorem~\ref{theorem:existence} states that there is no loss of optimality in restricting attention to Markovian deterministic policies. Because histories need not be retained and nonrandomized policies suffice, this simplifies computation through reduced storage and fewer arithmetic operations. Importantly, by embedding the action-dependent dwell times into the transition probabilities and costs, the \emph{epoch-by-epoch} SMDP is transformed into an equivalent \emph{slot-by-slot} MDP. Hence, classical dynamic programming methods developed for MDPs, such as relative value iteration (RVI), can be applied to solve~\eqref{eq:bellman-equation}.

The RVI operates as follows. For each iteration $n \geq 1$, it updates the value function using the recursion
\begin{subequations}
\begin{align}
    Q^{n}(s, u) &= c(s, u) + \sum\nolimits_{s^{\prime}}P_{s,s^{\prime}}(u) V^{n-1}(s^{\prime}),\label{eq:rvi-a}\\
    \tilde{V}^{n}(s) &= \min_{u \in \mathcal{U}} \left[Q^{n}(s, u)\right],\label{eq:rvi-b}\\
    V^{n}(s) &= \tilde{V}^{n}(s) - \tilde{V}^{n}(s_{\textrm{ref}}),\label{eq:rvi-c}
\end{align}\label{eq:rvi}
\end{subequations}
where $Q^{n}(s,u)$ is the Q-factor at iteration $n$ and $s_{\textrm{ref}}$ is an arbitrary reference state. The sequences $\{\tilde{V}^{n}(s)\}$ and $\{V^{n}(s)\}$ converge as $n \to \infty$~\cite{puterman1994markov}. Moreover, we have $\lambda^{*} = \tilde{V}(s_{\textrm{ref}})$ and $V(s) = \tilde{V}(s) - \tilde{V}(s_{\textrm{ref}})$ as a solution to~\eqref{eq:bellman-equation}. 

For numerical tractability, RVI is implemented on a finite state space. That is, the age processes are truncated at a sufficiently large bound $N$~\cite{Shi2020TCNS, luo2024cost}, i.e.,
\begin{align*}
    \tilde{\delta}_t = \min\{\delta_{t}, N\}, 
    ~ \tilde{\tau}_{t} = \min\{\tau_{t}, N\}.
\end{align*}
Let $\mathcal{S}_{N} = \{(\tau, \delta): 1 \leq \tau, \delta \leq N\}$ denote the truncated state space. The per-iteration time complexity of RVI is $\mathcal{O}(|\mathcal{S}_{N}|^2| \mathcal{U}|)$, which becomes computationally prohibitive when $N$ is large~\cite{puterman1994markov}. In the sequel, we establish useful monotonicity properties of the optimal policy and propose low-complexity, structure-aware algorithms.

\subsection{Monotonicity Properties}\label{sec:structure}
We now present structural results on the optimal policy. As will become apparent, the introduction of channel renewal makes our problem challenging and disrupts some of the convenient properties that would hold in the absence of renewal.

We first show the monotonicity of the value function $V$. 

\begin{proposition}\label{proposition:monotonicity-of-V}
The value function $V$ given in~\eqref{eq:bellman-equation} satisfies:
\begin{itemize}
    \item[i.] for any $\tau\in \mathbb{N}^{+}$, $V(\tau, \delta)$ is increasing in $\delta$.
    \item[ii.] for any $\delta\in \mathbb{N}^{+}$, $V(\tau, \delta)$ is increasing in $\tau$.
\end{itemize}
\end{proposition}
\begin{proof}
See Appendix~\ref{proof:proposition-monotonicity-V}.
\end{proof}

These properties align with the intuition that starting with fresh information, or with a good channel, yields a lower minimum cost than starting with outdated information, or with a poor channel. As channel quality deteriorates and information ages, it can become beneficial to take more aggressive actions, such as increasing transmission frequency or restoring channel quality. To formalize these properties, we introduce the following definitions.

\begin{definition}[AoI-monotone]
    A policy $\pi$ is AoI-monotone if, for any fixed $\tau$, it satisfies $\pi(\tau, \delta^{\prime}) \geq \pi(\tau, \delta)$ for all $\delta^{\prime} \geq \delta$.
\end{definition}
\begin{definition}[AoC-monotone]\label{def:AoC-monotone}
    A policy $\pi$ is AoC-monotone if, for any fixed $\delta$, it satisfies $\pi(\tau^{\prime}, \delta) \geq \pi(\tau, \delta)$ for all $\tau^{\prime} \geq \tau$.
\end{definition}

Establishing structural results aligned with these definitions requires the submodularity property of the Q-factor.
\begin{definition}[Submodularity]\label{definition:submularity} A function $g(x, y)$ is submodular on $\mathcal{X}\times\mathcal{Y}$, if for all $x^{\prime} \geq x$ and $y^{\prime} \geq y$,
\begin{equation*}
    g(x^{\prime}, y^{\prime}) + g(x, y) \leq g(x^{\prime}, y)+g(x, y^{\prime})
\end{equation*}
If the inequality is reversed, $g(x,y)$ is called supermodular.
\end{definition}

\begin{lemma}[{\hspace{-0.03em}\cite[Lemma~4.7.1]{puterman1994markov}}]
\label{lemma:submodularity}
If the Q-factor $Q(s, u)$ is submodular in $(\tau, u)$ for each fixed $\delta$, then
\begin{equation*}
    \pi^{*}(\tau, \delta) = \argmin_{u\in\mathcal{U}}Q(s, u)
\end{equation*}
is AoC-monotone. Similarly, if the Q-factor $Q(s, u)$ is submodular in $(\delta,u)$ for each fixed $\tau$, then $\pi^{*}(\tau, \delta)$ is AoI-monotone.
\end{lemma}

This lemma can be leveraged to conclude about the monotonicity of the optimal policy owing to the following result.

\begin{theorem}\label{theorem:submodular-reduced-space}
The optimal policy has the following monotonicity properties:
\begin{enumerate}
    \item[i.] $Q(s, u)$ is submodular in $(\tau, u)$ for $u\in\mathcal{U}$; that is, the optimal policy is AoC-monotone.
    \item[ii.] $Q(s, u)$ is submodular in $(\delta, \underline{u})$ for $\underline{u}\in\{0, 1\}$. 
\end{enumerate}
\end{theorem}
\begin{proof}
    See Appendix~\ref{proof:submodularity-Q}.
\end{proof}

\begin{figure}[t!]
    \centering
    \includegraphics[width=\linewidth]{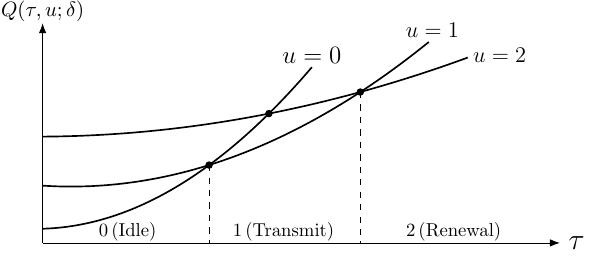}
    \caption{An illustration of AoC-monotonicity of the optimal policy.}
    \label{fig:monotone-policy}
\end{figure}

The AoC-monotonicity of an optimal policy is illustrated in Fig.~\ref{fig:monotone-policy}. For each $u$, the Q-factor $Q(\tau, u;\delta)$ represents the cost of taking an action $u$ in the current decision epoch and following the optimal policy onwards. The optimal action $u^{*}$ is the one that minimizes the Q-factor, and the value function is obtained as $V(s) = Q(s,u^{*})$.

Note, however, that the optimal policy is not necessarily AoI-monotone. This is due to the possibility of taking renewal actions. To see this, let $u^{\prime}=2$ and $u \in \{0, 1\}$. We have
\begin{align*}
    c(\delta^{\prime}, u^{\prime};\tau)
    & + c(\delta, u;\tau) - c(\delta^{\prime}, u;\tau)-c(\delta, u^{\prime};\tau)\\
    &=\frac{1}{\delta_{\textrm{R}}} \sum_{i=1}^{\delta_{\textrm{R}}-1}
    \left(f(\delta^{\prime} + i) - f(\delta+i)\right) \geq 0
\end{align*}
for all $\delta^{\prime} \geq \delta$, which implies supermodularity of the cost function $c(\delta, u;\tau)$. Therefore, we cannot say anything specific about the Q-factor $Q(\delta, u;\tau)$, as it combines both submodular and supermodular components. 

\subsection{Structured-Aware Algorithms}\label{sec:algorithm}
This section exploits the AoC-monotonicity of the optimal policy to develop a \emph{Structured Policy Iteration} (SPI) algorithm with reduced computational complexity. The key idea is that, for each AoI value $\delta$, instead of determining the optimal action for every AoC value $\tau$, it suffices to identify two threshold values: one for transmission and one for renewal. The SPI algorithm proceeds as follows:
\begin{itemize}
    \item[(1)] \textit{Initialization:} Select an initial policy $\pi^{0}$, a reference state $s_{\textrm{ref}}$. 
    \item[(2)] \textit{Policy Evaluation:} For each $n = 0, 1, 2, \ldots$, obtain $\lambda^{n}$ and $V^{n}$ by solving the Bellman equation
    \begin{equation*}
        \lambda^{n}
        + V^{n}(s) = c(s, \pi^{n}(s)) + \sum\nolimits_{s^{\prime}} P_{s,s^{\prime}}(\pi^{n}(s))V^n(s^{\prime})
    \end{equation*}
    for all $s \in \mathcal{S}_N$ such that $V^{n}(s_{\textrm{ref}}) = 0$.
    \item[(3)] \textit{Policy Improvement:} For each AoI value $\delta$, $1\leq \delta\leq N$, update $\pi^{n+1}(\tau, \delta)$ in the increasing order of the AoC $\tau$:
    \begin{itemize}
        \item[a.] Initialize $s = (\tau, \delta)$ with $\tau = 1$.
        \item[b.] Update $\pi^{n+1}(s)$ as
        \begin{equation*}
            \pi^{n+1}(s)= \argmin_{u\in\mathcal{U}}=
            \left[c(s, u) + \sum\nolimits_{s^{\prime}} P_{s,s^{\prime}}(u)V^{n}(s^{\prime})\right].
        \end{equation*}
        \item[c.] If $\pi^{n+1}(s)=1$, then the optimal action for all subsequent states $s^{\prime} = (\tau^{\prime}, \delta)$ with $\tau^{\prime}>\tau$ is either to transmit or to restore. Moreover, if $\pi^{n+1}(s)=2$, then $\pi^{n+1}(s^{\prime})=2$ for all $\tau^{\prime}>\tau$.
    \end{itemize}
    \item[(4)] \textit{Stopping Criterion:} If $\pi^{n+1} = \pi^{n}$, the algorithm terminates with $\lambda^{*} = \lambda^{n}$ and $\pi^{*} = \pi^{n}$; otherwise increase $n=n+1$ and return to Step~(2).
\end{itemize}

The worst-case per-iteration time complexity of SPI is $\mathcal{O}(N^{2}|\mathcal{U}|)$, which is significantly lower than the $\mathcal{O}(|\mathcal{S}_N|^{2}| \mathcal{U}|)$ time complexity of the unstructured RVI algorithm. However, Step~(3) can still be time-consuming, since the search must continue until the renewal threshold is identified.

We next propose a suboptimal \emph{modified SPI} algorithm with a reduced time complexity of $\mathcal{O}(N^{2})$. The modified SPI operates as follows. First, the channel is renewed once its quality degrades below a threshold $\theta_{\textrm{th}}$, or equivalently, when the AoC reaches a threshold $\tau_{\max}$ such that $\theta_{\textrm{th}} = \theta(\tau_{\max})$. Then, in Step~(3) of SPI, determine a transmission threshold $\tau_{\delta, \textrm{th}}$ for each given $\delta$, and set $\pi^{n+1}(\tau, \delta) = 1$ for all $\tau_{\delta, \textrm{th}} < \tau < \tau_{\max}$ without further computation.

\begin{figure*}[t!]
\centering
\subfloat[SPI and RVI ($\alpha = 0.05$)]{
\includegraphics[width=0.31\linewidth]{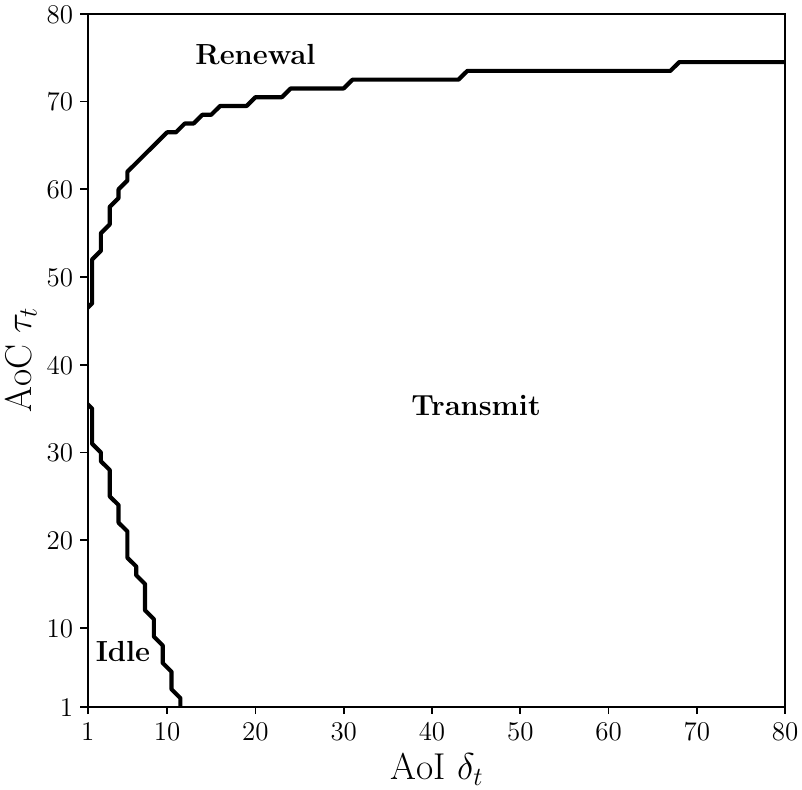}
\label{fig:spi}}
\subfloat[modified SPI ($\alpha = 0.05$)]{
\includegraphics[width=0.31\linewidth]{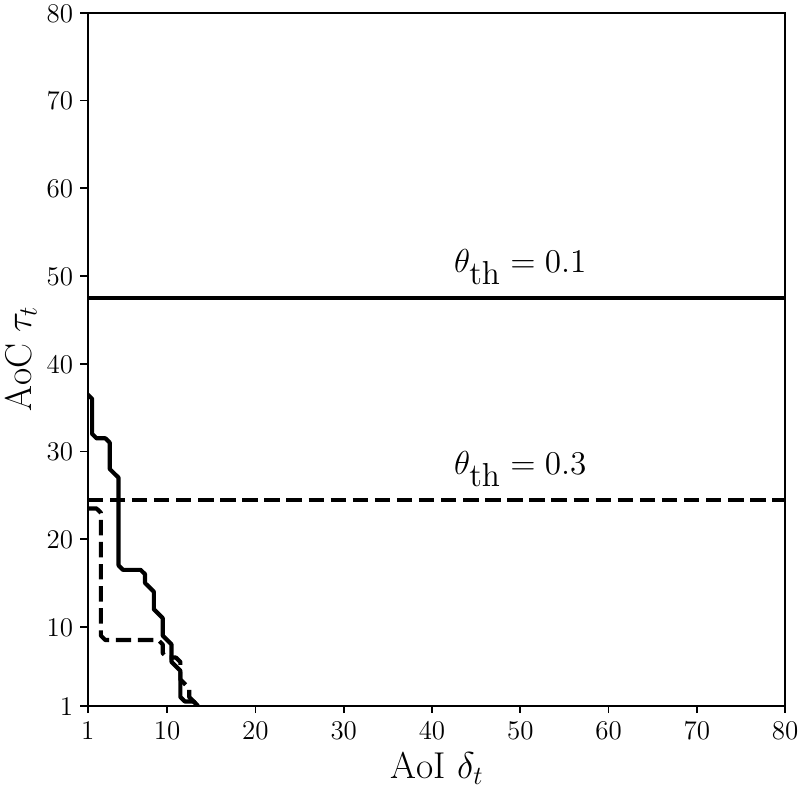}
\label{fig:modified SPI}
}
\subfloat[Performance comparison ($\theta_\text{th} = 0.1$)]{
\centering
\includegraphics[width=0.365\linewidth]{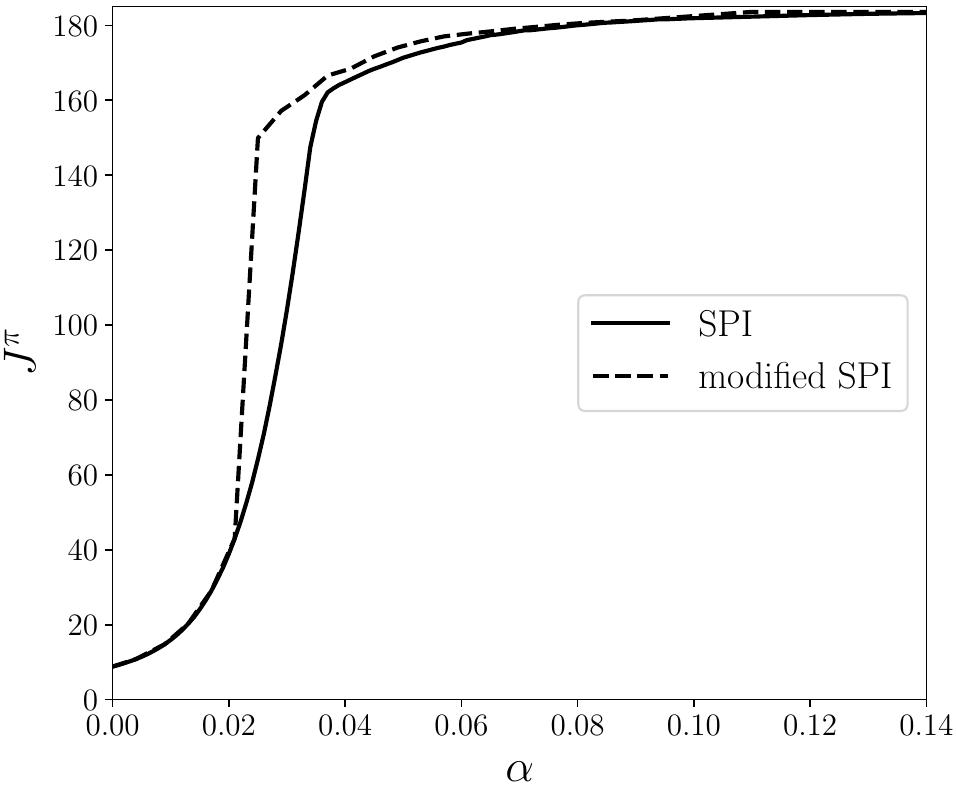}
\label{fig:performance comparison}
}
\caption{Behavior and performance comparisons of the proposed algorithms.}
\label{fig:algorithms}
\end{figure*}

\section{Numerical Example}\label{sec:numerical-results}
We examine a remote estimation system with parameters
\begin{equation*}
    \mathbf{A} = \begin{bmatrix}
        1.1 & 0.5\\
        0 & 0.8
    \end{bmatrix}, \mathbf{Q} = \begin{bmatrix}
        1 & 0\\
        0 & 1
    \end{bmatrix}, \mathbf{C} = [1, 1], ~\mathbf{R} = 1.
\end{equation*}
The channel is modeled as an exponential wearing process
\begin{equation*}
    \theta(\tau) = (\theta_{\max}-\theta_{\min}) e^{- \alpha \tau} + \theta_{\min},
\end{equation*}
where $\alpha \geq 0$ is the decay rate. The case $\alpha = 0$ corresponds to an \emph{i.i.d.} channel with constant success probability $\theta_{\max}$. For $\alpha > 0$, the channel reliability $\theta(\tau)$ decreases from $\theta_{\max}$ to $\theta_{\min}$ as the AoC goes to infinity. Larger values of $\alpha$ correspond to faster channel deterioration. In the simulations, we set $\theta_{\max} = 0.99$ and $\theta_{\min} = 0$. The transmission and renewal energy consumptions are $E_{\textrm{T}} = 10$ and $E_{\textrm{R}} = 100$, respectively. Each transmission incurs $\tau_{\textrm{D}} = 15$ units of channel wear, each renewal action takes $\delta_{\textrm{R}} = 30$ time slots, and the age processes are truncated at $N = 100$ when implementing the RVI, SPI, and modified SPI algorithms. When testing performance, we take the average cost over $10^{7}$ time slots.

Fig.~\ref{fig:spi} and Fig.~\ref{fig:modified SPI} compare the behavior of the policies obtained by different algorithms. It is observed that the optimal policy obtained by SPI (and RVI) is AoC-monotone but not AoI-monotone. This aligns with the intuition that, when the channel condition is good and the information at the receiver is relatively fresh, the sensor may remain idle to save energy. In contrast, as the channel quality deteriorates and the information becomes stale, the sensor takes more aggressive actions to improve both estimation quality and channel reliability. On the other hand, the modified SPI specifies a constant renewal threshold $\theta_\text{th}$ for all AoI values. As $\theta_\text{th}$ increases, the worst channel condition improves, and consequently, the sensor transmits more frequently to improve estimation quality.

Fig.~\ref{fig:performance comparison} shows the estimation performance achieved by SPI and modified SPI under different decay rates $\alpha$. The modified SPI achieves near-optimal performance for both slowly and rapidly wearing channels. However, noticeable performance degradation occurs for moderate values of $\alpha$, where the timing of channel renewal has a significant impact on system performance. This highlights the need for careful calibration of renewal thresholds in such cases.

\section{Conclusions}\label{sec:conclusion}
We studied optimal transmission scheduling for remote estimation over a wearing channel whose reliability deteriorates with time and usage. This setting creates a fundamental tradeoff between information freshness and channel aging. We formulated the problem as an SMDP characterized by two dependent age processes. We established the existence of an AoC-monotone optimal policy and developed a structured policy iteration algorithm to determine the optimal timing of data transmission and channel renewal. Numerical results validated the proposed structural findings and algorithmic design.

\section*{Acknowledgment}
The authors would like to thank the anonymous reviewers for their valuable feedback, in particular the reviewer who pointed out an error in the initial submission.

\appendices

\section{Proof of Theorem~\ref{theorem:existence}}\label{proof:theorem-existence} 
The equivalence between the SMDP $(\mathcal{S}, \mathcal{U}, \tilde{P}, \tilde{c}, \ell)$ and the MDP $(\mathcal{S}, \mathcal{U}, P, c)$ follows from the following result. 

\begin{lemma}[{\hspace{-0.03em}\cite[Proposition~11.4.5]{puterman1994markov}}]
Let $\gamma$ be a scalar such that $0 < \gamma < \ell(u) / (1-\tilde{P}_{s,s}(u))$ for all $s\in\mathcal{S}$ and $u\in\mathcal{U}$ for which $\tilde{P}_{s,s}(u) < 1$. Define for all $s$ and $u$,
\begin{equation*}
    c(s, u) = \frac{\tilde{c}(s,u)}{\ell(u)},\,P_{s,s^{\prime}}(u) = \begin{cases}
        \frac{\gamma \tilde{P}_{s,s^{\prime}}(u)}{\ell(u)}, &s \neq s^{\prime},\\
        1-\frac{\gamma \left(1-\tilde{P}_{s,s}(u) \right)}{\ell(u)}, &s=s^{\prime}.
    \end{cases}
\end{equation*}
If there exists a scalar $\lambda^{*}$ and a vector $V$ such that for all $s$,
\begin{equation*}
    V(s) = \min_{u\in \mathcal{U}} \left[c(s,u) - \lambda^{*} + \sum\nolimits_{s^{\prime}} P_{s,s^{\prime}}(u)V(s^{\prime})
    \right],
\end{equation*}
then $(\lambda^{*}, \tilde{V})$ satisfies $\tilde{V}= \gamma V$ and
\begin{equation*}
    \tilde{V}(s) = \min_{u} \left[
    \tilde{c}(s,u) - \lambda^{*} \ell(u) + \sum\nolimits_{s^{\prime}} \tilde{P}_{s,s^{\prime}}(u) \tilde{V}(s^{\prime}) \right].
\end{equation*}
\end{lemma}

Setting $\gamma=1$ proves the equivalence result in Theorem~\ref{theorem:existence}. 

The key step is then to show that there exists a stationary policy $\pi$ such that $V(s)$ is finite for all $s$. We rely on the vanishing discount approach~\cite{Onesimo1996MDP} to establish this result. The $\alpha$-discounted Bellman equation is defined as
\begin{align}
     V_{\alpha}(s) := \min_{u} \left[
     c(s, u) + \alpha \sum\nolimits_{s^{\prime}}P_{s,s^{\prime}}(u) V_\alpha(s^{\prime})
     \right]. \label{eq:discounted-bellman-eq}
\end{align}
For a reference state $s_{\textrm{ref}} \in \mathcal{S}$, define
$h_{\alpha}(s) = V_{\alpha}(s) - V_{\alpha}(s_{\textrm{ref}})$. Then,
\begin{equation*}
    (1 -\alpha) V_{\alpha}(s_{\textrm{ref}}) + h_{\alpha}(s) =
    \min_{u} \left[
     c(s, u) + \alpha \sum\nolimits_{s^{\prime}} P_{s,s^{\prime}}(u) h_{\alpha}(s^{\prime})
     \right].
\end{equation*}
Taking the limit as $\alpha\uparrow 1$, and assuming that the limits of all terms in the above equation exist, the Bellman equation for the original average-cost problem satisfies
\begin{equation}
    \lambda^{*} = \lim_{\alpha \uparrow 1}(1-\alpha)V_{\alpha}(s_{\textrm{ref}}), \quad V(s) = \lim_{\alpha \uparrow 1} h_{\alpha}(s).\label{eq:limits}
\end{equation}

To establish the convergence of the limits in~\eqref{eq:limits}, it suffices to verify the following conditions~\cite[Theorem~5.5.4]{Onesimo1996MDP}:
\begin{itemize}
    \item[i.] $V_{\alpha}(s) < \infty$ for every $s \in \mathcal{S}$ and $\alpha \in (0, 1)$.
    \item[ii.] There exists a reference state $s_{\textrm{ref}}$, constants $M \geq 0$ and $\underline{\alpha}\in (0, 1)$, and a nonnegative function $b(s)$ such that $- M \leq h_{\alpha}(s) \leq b(s)$ for all $s \in \mathcal{S}$ and $\alpha \in[\underline{\alpha}, 1)$.
    \item[iii.] $b(s)$ satisfies $\sum_{s^{\prime}} b(s) P_{s,s^{\prime}}(u) < \infty$ for all $s$ and $u$.
\end{itemize}

Condition~(i) ensures the existence of a unique solution to~\eqref{eq:discounted-bellman-eq}, while Conditions~(ii)--(iii) guarantee the convergence of the limits in~\eqref{eq:limits}. We will use the following lemma.

\begin{lemma}[{\hspace{-0.03em}
\cite[Lemma~5.3.1]{Onesimo1996MDP}}]\label{lemma:discounted-value}
The $\alpha$-discounted value function satisfies for every $s$, $\lim_{\alpha \uparrow 1}(1 - \alpha) V_{\alpha}(s) \leq J^{*}(s).$
\end{lemma}

This lemma states that the normalized $\alpha$-discounted value function provides a lower bound on the average-cost criterion. Together with Lemma~\ref{lemma:with-renewal}, it guarantees that Conditions~(i)--(iii) hold.

\section{Proof of Proposition~\ref{proposition:monotonicity-of-V}}\label{proof:proposition-monotonicity-V}
Recall that $V(s) = \lim_{n \to \infty}V^{n}(s)$, where the sequence $\{V^{n}\}$ is generated by the RVI recursion~\eqref{eq:rvi}. This suggests that we can prove part (i) by induction. Choose $V^{0}(\tau, \delta)$ to be increasing in $\delta$, i.e.,
\begin{equation*}
    V^{0}(\tau, \delta) \leq V^{0}(\tau, \delta^{\prime}),~ \forall \delta \leq \delta^{\prime}.
\end{equation*}

Let us now assume (the induction hypothesis) that $V^{n}(\tau, \delta)$ is increasing in $\delta$ for all $\tau$. We will show that $V^{n+1}(\tau, \delta)$ is increasing in $\delta$. From the RVI recursion~\eqref{eq:rvi}, we have
\begin{equation*}
    V^{n+1}(s) = \min_{u}\left[Q^{n+1}(s,u)\right] - \min_{u}\left[Q^{n+1}(s_{\textrm{ref}},u)\right].
\end{equation*}
Hence, it suffices to show that $Q^{n+1}(s, u)$ is increasing in $\delta$ for all $u$. We establish this as follows. 

If $u=0$, the system jumps into $(\tau + 1, \delta + 1)$ with certainty. Then
\begin{equation*}
    Q^{n+1}(\delta, 0;\tau) = f(\delta) + V^{n}(\tau + 1, \delta + 1),
\end{equation*}
which by Lemma~\ref{lemma:monotonicity-of-AoI} and the induction hypothesis gives
\begin{align}
    Q^{n+1}(\delta, 0;\tau) 
    &\leq f(\delta^{\prime}) + V^{n}(\tau + 1, \delta^{\prime} + 1) \notag\\
    &=Q^{n+1}(\delta^{\prime}, 0;\tau).\label{eq:prove-u0}
\end{align}
If $u=1$, the system can either transition to state $(\tau + \tau_{\textrm{D}}, 1)$ w.p. $\theta(\tau)$ or to state $(\tau + \tau_{\textrm{D}}, \delta + 1)$ w.p. $\bar{\theta}(\tau)$. Then, 
\begin{align}
      &Q^{n+1}(\delta, 1;\tau) \notag\\
    =~& f(\delta) + E_{\textrm{T}} + \theta(\tau) V^{n}(\tau + \tau_{\textrm{D}}, 1) + \bar{\theta}(\tau) V^{n}(\tau + \tau_{\textrm{D}}, \delta + 1) \notag\\
    \leq~& f(\delta^{\prime}) + E_{\textrm{T}} + \theta(\tau) V^{n}(\tau + \tau_{\textrm{D}}, 1) +
    \bar{\theta}(\tau) V^{n}(\tau +
    \tau_{\textrm{D}}, \delta^{\prime} + 1) \notag\\
    =~& Q^{n+1}(\delta^{\prime}, 1;\tau) \leq 0.\label{eq:prove-u1}
\end{align}
If $u=2$, the system moves to $(1, \delta + \delta_{\textrm{R}})$ w.p. $\frac{1}{\delta_{\textrm{R}}}$ or remains in $(\tau, \delta)$ w.p. $1-\frac{1}{\delta_{\textrm{R}}}$. Thus, we have
\begin{align}
    &Q^{n+1}(\delta, 2;\tau) \notag\\
    =& \frac{\sum_{i=0}^{\delta_{\textrm{R}}-1} f(\delta+i) + E_{\textrm{R}}}{\delta_\textrm{R}} + \frac{V^{n}(1, \delta + \delta_{\textrm{R}})}{\delta_\textrm{R}}+\frac{(\delta_{\textrm{R}}-1)V^{n}(\tau, \delta)}{\delta_\textrm{R}} \notag\\
    \leq& \frac{\sum_{i=0}^{\delta_\textrm{R}-1} f(\delta^{\prime} +i) + E_\textrm{R}}{\delta_\textrm{R}} + \frac{V^{n}(1, \delta^{\prime
    } + \delta_\textrm{R})}{\delta_\textrm{R}}+\frac{(\delta_\textrm{R} - 1)V^{n}(\tau, \delta^\prime)}{\delta_\textrm{R}} \notag\\
    =& Q^{n+1}(\delta^{\prime}, 2;\tau). \label{eq:prove-u2}
\end{align}
It follows from~\eqref{eq:prove-u0}--\eqref{eq:prove-u2} that $V^{n+1}(s)$ is increasing in $\delta$.

We adopt a similar induction procedure to prove part (ii). Suppose that $V^{n}(\tau, \delta)$ is increasing in $\tau$ for all $\delta$. We need to show that $Q^{n+1}(\tau, u;\delta)$ is increasing in $\tau$ for all $u$. Obviously, $Q^{n+1}(\tau, 0;\delta)$ and $Q^{n+1}(\tau, 2;\delta)$ are increasing in $\tau$. For $u=1$, we have
\begin{align*}
    &Q^{n+1}(\tau^{\prime}, 1;\delta) - Q^{n+1}(\tau, 1;\delta)  \\
    =~& \theta(\tau^{\prime})V^{n}(\tau^{\prime} + \tau_{\textrm{D}}, 1) - \theta(\tau)V^{n}(\tau + \tau_{\textrm{D}}, 1)  \\
    &+ \bar{\theta}(\tau^{\prime})V^{n}(\tau^{\prime} + \tau_{\textrm{D}}, \delta + 1) - \bar{\theta}(\tau)V^{n}(\tau + \tau_{\textrm{D}}, \delta+ 1)\notag\\
    =~& \theta(\tau^{\prime})\big(
    \underbrace{V^{n}(\tau^{\prime} + \tau_{\textrm{D}}, 1) - V^{n}(\tau + \tau_{\textrm{D}}, 1)}_{\geq 0}
    \big) \\
    \quad& + \bar{\theta}(\tau^{\prime})\big(
    \underbrace{V^{n}(\tau^{\prime} + \tau_{\textrm{D}}, \delta + 1) - V^{n}(\tau + \tau_{\textrm{D}}, \delta + 1)}_{\geq 0}
    \big) \\
    &+ (\underbrace{
    \theta(\tau^{\prime}) -  \theta(\tau)}_{\leq 0})
     \big(
     \underbrace{V^{n}(\tau + \tau_{\textrm{D}}, 1) - V^{n}(\tau + \tau_{\textrm{D}}, \delta + 1 )}_{\leq 0}
     \big) 
\end{align*}
for all $\tau^{\prime} \geq \tau$, which concludes our inductive proof.

\section{Proof of Theorem~\ref{theorem:submodular-reduced-space}}\label{proof:submodularity-Q}
We first show that $Q(\delta,u;\tau)$ is submodular in $(\delta, u)$ for $u\in\{0, 1\}$. By Definition~\ref{definition:submularity}, we need to verify that
\begin{equation}
    Q(\delta^{\prime}, 1;\tau) + Q(\delta, 0;\tau) - 
    Q(\delta^{\prime}, 0;\tau) -Q(\delta, 1;\tau) \leq 0, \label{eq:prove-submodularity-init}
\end{equation}
Substituting
\begin{align*}
    Q(\delta, 0;\tau) &= f(\delta) + V(\tau + 1, \delta+1),\\
    Q(\delta, 1;\tau) &= f(\delta) \hspace{-0.2em}+\hspace{-0.2em} E_{\textrm{T}} \hspace{-0.2em}+\hspace{-0.2em} \theta(\tau)V(\tau \hspace{-0.2em}+\hspace{-0.2em} \tau_{\textrm{D}}, 1) \hspace{-0.2em}+\hspace{-0.2em} \bar{\theta}(\tau)V(\tau \hspace{-0.2em}+\hspace{-0.2em} \tau_{\textrm{D}}, \delta\hspace{-0.2em}+\hspace{-0.2em}1)
\end{align*}
into~\eqref{eq:prove-submodularity-init} gives the following inequality
\begin{align}
    \psi(\tau, \delta, \delta^{\prime}) :=~& \bar{\theta}(\tau) \big(
        V(\tau+\tau_{\textrm{D}}, \delta^{\prime}) - V(\tau+\tau_{\textrm{D}}, \delta)
    \big) \notag\\
    &-\big(V(\tau+1, \delta^{\prime}) - V(\tau+1, \delta)\big) \leq 0. \label{eq:prove-submodularity-condition}
\end{align}
By Proposition~\ref{proposition:monotonicity-of-V}, $V(\tau, \delta)$ is increasing in $\delta$ and $\tau$. So it is not immediately clear whether the inequality in \eqref{eq:prove-submodularity-condition} holds in general. Next, we will derive an upper bound of $\psi(\tau, \delta, \delta^{\prime})$, denoted as $\Psi(\tau, \delta, \delta^{\prime})$, and show that $\Psi(\tau, \delta, \delta^{\prime})\leq 0$. 

The upper bound is derived as follows. Let $u_{1}^{*}, u_{2}^{*} \in \{0, 1\}$ denote the optimal actions such that
\begin{align*}
    V(\tau+\tau_{\textrm{D}}, \delta) 
    &=\min_{u} Q(\delta, u;\tau+\tau_{\textrm{D}}) = Q(\delta, u_{1}^{*};\tau+\tau_{\textrm{D}}),\\
    V(\tau+1, \delta^{\prime}) &=\min_{u} Q(\delta, u;\tau+1) = Q(\delta^{\prime}, u_{2}^{*}; \tau+1). 
\end{align*}
Letting $u_{0} = u_{1}^{*} $ and $u_{3}=u_{2}^{*}$ gives that
\begin{align*}
    V(\tau+\tau_{\textrm{D}}, \delta^{\prime}) 
    &= \min_{u} Q(\delta^{\prime}, u;\tau+\tau_{\textrm{D}}) 
    \leq Q(\delta^{\prime}, u_{0};\tau+\tau_{\textrm{D}}),\\
    V(\tau+1, \delta) &= \min_{u} Q(\delta, u;\tau+1) \leq Q(\delta, u_{3};\tau+1).
\end{align*}
Substituting the above inequalities into~\eqref{eq:prove-submodularity-condition} yields
\begin{align}
    \psi(\tau, \delta, \delta^{\prime}) &\leq \bar{\theta}(\tau) \big(
        Q(\delta^{\prime}, u_{0};\tau+\tau_{\textrm{D}}) - Q(\delta, u_{1}^{*};\tau+\tau_{\textrm{D}}) \big) \notag\\
    &\quad\quad-\big(Q(\delta^{\prime}, u_{2}^{*};\tau+1) - Q(\delta, u_{3};\tau+1)\big)\notag\\
    &=: \Psi(\tau, \delta, \delta^{\prime}). \label{eq:prove-structure-1}
\end{align}
The proof proceeds by induction. Assume that $V^{n}(\tau, \delta)$ satisfies~\eqref{eq:prove-submodularity-condition} for all $\tau$ and $\delta^{\prime} \geq \delta$. This is equivalent to assume that $Q^{n+1}(\delta,u;\tau)$ is submodular at iteration $n+1$. We then show that $\Psi^{n+1}(\tau, \delta, \delta^{\prime})\leq 0$ for all $\tau$ and $\delta^{\prime} \geq \delta$.

The proof is divided into four cases: (1) $u_{1}^{*} = u_{2}^{*} = 0$; (2) $u_{1}^{*} = 1$, $u_{2}^{*} = 0$; (3) $u_{1}^{*} = 0$, $u_{2}^{*} = 1$; and (4) $u_{1}^{*} = 1$, $u_{2}^{*} = 1$. 

Case 1): When $u_{1}^{*} = u_{2}^{*} = 0$, we have
\begin{align*}
    &Q^{n+1}(\delta^{\prime}, u_{0}; \tau+\tau_{\textrm{D}}) - Q^{n+1}(\delta, u_{1}^{*};\tau + \tau_{\textrm{D}}) \notag\\
    =& f(\delta^{\prime}) \hspace{-0.2em}-\hspace{-0.2em} f(\delta)
    \hspace{-0.2em}+\hspace{-0.2em} V^{n}(\tau \hspace{-0.2em}+\hspace{-0.2em} \tau_{\textrm{D}}\hspace{-0.2em} +\hspace{-0.2em} 1, \delta^{\prime} \hspace{-0.2em}+\hspace{-0.2em}1) \hspace{-0.2em}-\hspace{-0.2em}  V^{n}(\tau \hspace{-0.2em}+\hspace{-0.2em}\tau_{\textrm{D}} \hspace{-0.2em}+\hspace{-0.2em} 1, \delta \hspace{-0.2em}+\hspace{-0.2em} 1),\\
    &Q^{n+1}(\delta^{\prime}, u_{2}^{*};\tau+1) - Q^{n+1}(\delta, u_{3};\tau+1) \\
    =& f(\delta^{\prime}) - f(\delta) + V^{n}(\tau+2, \delta^{\prime}+1) -  V^{n}(\tau+2, \delta+1).
\end{align*}
Let $\tau_{0} = \tau + 1$. Substituting the above equalities into~\eqref{eq:prove-structure-1} gives
\begin{align}
    \Psi^{n+1}&(\tau, \delta, \delta^{\prime}) 
    = -\theta(\tau) \big( f(\delta^{\prime}) - f(\delta)\big) \notag\\
    &+ \Big( \bar{\theta}(\tau) \big(
        V^{n}(\tau_0+\tau_{\textrm{D}}, \delta^{\prime} + 1) -  V^{n}(\tau_0 + \tau_{\textrm{D}}, \delta + 1)
    \big) \notag\\
    &- \big( V^{n} (\tau_{0} + 1, \delta^{\prime} +1) -  V^{n}(\tau_{0} + 1, \delta + 1)
    \big)\Big)\notag\\
    &\overset{(a)}{\leq} -\theta(\tau) \big(f(\delta^{\prime}) - f(\delta) \big) \notag\\
    &+\Big(\bar{\theta}(\tau_{0}) \big(
        V^{n}(\tau_{0} +\tau_{\textrm{D}}, \delta^{\prime} + 1) - V^{n}(\tau_{0}+ \tau_{\textrm{D}}, \delta + 1)
    \big) \notag\\
    &-\big(
        V^{n}(\tau_{0} + 1, \delta^{\prime} + 1) -  V^{n}(\tau_{0} + 1, \delta + 1)
    \big)\Big) \notag \\
    &= - \theta(\tau) \big(f(\delta^{\prime}) -f(\delta)\big) + \psi^{n}(\tau_{0}, \delta, \delta^{\prime}) \overset{(b)}{\leq} 0, \label{eq:prove-thm-1}
\end{align}
where $(a)$ follows because $\bar{\theta}(\tau)$ is increasing in $\tau$, and $(b)$ follows form the induction hypothesis.

Case 2): $u_{1}^{*} = 1, u_{2}^{*} = 0$. We may write
\begin{align}
    \hspace{-0.4em}\Psi^{n+1}(\tau, \delta, \delta^{\prime}) =& -\theta(\tau) \big(f(\delta^{\prime}) - f(\delta)\big) + \bar{\theta}(\tau) \bar{\theta}(\tau+\tau_{\textrm{D}})\notag\\
    &\times \big(
    V^{n}(\tau+2\tau_{\textrm{D}}, \delta^{\prime} \hspace{-0.2em}+\hspace{-0.2em} 1) \hspace{-0.2em}-\hspace{-0.2em}  V^{n}(\tau\hspace{-0.2em}+ \hspace{-0.2em} 2\tau_{\textrm{D}}, \delta\hspace{-0.2em}+\hspace{-0.2em} 1) \big)\notag\\
    &-\big( V^{n}(\tau\hspace{-0.2em}+\hspace{-0.2em} 2, \delta^{\prime} \hspace{-0.2em} +\hspace{-0.2em} 1) \hspace{-0.2em}-\hspace{-0.2em}  V^{n}(\tau\hspace{-0.2em}+\hspace{-0.2em} 2, \delta\hspace{-0.2em}+\hspace{-0.2em} 1)
    \big).\notag
\end{align}
Adding an auxiliary term on the right-hand side, we obtain
\begin{align}
    &\Psi^{n+1}(\tau, \delta, \delta^\prime) = -\theta(\tau) \big(f(\delta^\prime) - f(\delta)\big) \notag\\
    &\quad+\Big( \bar{\theta}(\tau) \bar{\theta}(\tau\hspace{-0.2em}+\hspace{-0.2em}\tau_\textrm{D})\big(
        V^{n}(\tau\hspace{-0.2em}+\hspace{-0.2em}2\tau_\textrm{D}, \delta^\prime\hspace{-0.2em}+\hspace{-0.2em}1) \hspace{-0.2em}-\hspace{-0.2em}  V^{n}(\tau\hspace{-0.2em}+\hspace{-0.2em}2\tau_\textrm{D}, \delta\hspace{-0.2em}+\hspace{-0.2em}1)
    \big)\notag\\
    &\quad-\bar{\theta}(\tau) \big(
        V^{n}(\tau\hspace{-0.2em}+\hspace{-0.2em}\tau_\textrm{D}\hspace{-0.2em}+\hspace{-0.2em}1, \delta^\prime\hspace{-0.2em}+\hspace{-0.2em}1) \hspace{-0.2em}-\hspace{-0.2em}  V^{n}(\tau\hspace{-0.2em}+\hspace{-0.2em}\tau_\textrm{D}\hspace{-0.2em}+\hspace{-0.2em}1, \delta\hspace{-0.2em}+\hspace{-0.2em}1)
    \big)\Big)\notag\\
    & \quad + \Big(\bar{\theta}(\tau) \big(
        V^{n}(\tau\hspace{-0.2em}+\hspace{-0.2em}\tau_\textrm{D}\hspace{-0.2em}+\hspace{-0.2em}1, \delta^\prime\hspace{-0.2em}+\hspace{-0.2em}1) \hspace{-0.2em}-\hspace{-0.2em}  V^{n}(\tau\hspace{-0.2em}+\hspace{-0.2em}\tau_\textrm{D}\hspace{-0.2em}+\hspace{-0.2em}1, \delta\hspace{-0.2em}+\hspace{-0.2em}1)
    \big)\notag\\
    &\quad-\big(
        V^{n}(\tau+2, \delta^\prime+1) -  V^{n}(\tau+2, \delta+1)
    \big)\Big).\notag
\end{align}
Letting $\tau_{1} = \tau+\tau_{\textrm{D}}$ and $\tau_{2} = \tau + 1$ gives 
\begin{align}
\Psi^{n+1}(\tau, \delta, \delta^{\prime}) =& -\theta(\tau) \big(f(\delta^{\prime}) - f(\delta)\big) +\bar{\theta}(\tau) \psi^{n}(\tau_{1}, \delta, \delta^{\prime}) \notag\\
&+ \psi^{n}(\tau_{2}, \delta, \delta^{\prime}) \leq 0. \label{eq:prove-thm-2}
\end{align}

Case 3): $u_{0}^{*} = 0, u_{3}^{*} = 1$. This case follows by noting that 
\begin{align}
    &\Psi^{n+1}(\tau, \delta, \delta^{\prime}) = -\theta(\tau) \big(f(\delta^{\prime}) - f(\delta)\big) + \Big(
    \big(\bar{\theta}(\tau) - \bar{\theta}(\tau+1)\big) \notag\\
    &\times \big(V^{n}(\tau + \tau_{\textrm{D}} + 1, \delta^{\prime} + 1) - V^{n}(\tau + \tau_{\textrm{D}} + 1, \delta + 1)
    \big) \Big) \leq 0. \label{eq:prove-thm-3}
\end{align}

Case 4): $u_{0}^{*} = u_{3}^{*} = 1$. Let $\tau_{3} = \tau + \tau_{\textrm{D}}$. Then, we have
\begin{align}
    &\Psi^{n+1}(\tau, \delta, \delta^{\prime}) = - \theta(\tau) \big( f(\delta^{\prime}) \hspace{-0.2em}-\hspace{-0.2em} f(\delta)\big) \notag\\
    &\quad + \Big(\bar{\theta}(\tau) \bar{\theta}(\tau\hspace{-0.2em} +\hspace{-0.2em} \tau_{\textrm{D}}) \big( V^{n}(\tau\hspace{-0.2em}+\hspace{-0.2em}2 \tau_{\textrm{D}}, \delta^{\prime}\hspace{-0.2em}+\hspace{-0.2em} 1) \hspace{-0.2em}-\hspace{-0.2em}  V^{n}(\tau\hspace{-0.2em}+\hspace{-0.2em}2\tau_{\textrm{D}}, \delta\hspace{-0.2em} +\hspace{-0.2em} 1)
    \big) \notag\\
    &\quad - \big(
        V^{n}(\tau \hspace{-0.2em}+\hspace{-0.2em} \tau_{\textrm{D}} \hspace{-0.2em}+\hspace{-0.2em} 1, \delta^{\prime} \hspace{-0.2em}+\hspace{-0.2em} 1) \hspace{-0.2em}-\hspace{-0.2em} V^{n}(\tau \hspace{-0.2em}+\hspace{-0.2em} \tau_{\textrm{D}} \hspace{-0.2em} + \hspace{-0.2em} 1, \delta \hspace{-0.2em} + \hspace{-0.2em} 1)
    \big)\Big)\notag\\
    &\leq - \theta(\tau) \big(f(\delta^{\prime}) - f(\delta)\big) \notag\\
    &\quad + \bar{\theta}(\tau_{3})\big(
        V^{n}(\tau_{3} + \tau_{\textrm{D}}, \delta^{\prime} + 1) -  V^{n}(\tau_{3} + \tau_{\textrm{D}}, \delta + 1)
    \big)\notag\\
    &\quad - \big(
        V^{n}(\tau_{3} + 1, \delta^{\prime} + 1) - V^{n}(\tau_{3} + 1, \delta + 1)
    \big) \notag\\
    &= -\theta(\tau) \big(f(\delta^{\prime}) - f(\delta) \big) + \psi^{n}(\tau_{3}, \delta, \delta^{\prime}) \leq 0. \label{eq:prove-thm-4}
\end{align}

From~\eqref{eq:prove-thm-1}--\eqref{eq:prove-thm-4}, we conclude that $\Psi^{n+1}(\tau, \delta, \delta^{\prime})\leq 0$ holds for all possible $u_{1}^{*}, u_{2}^{*} \in\{0, 1\}$. It follows that $Q(\delta,u;\tau)$ is submodular on $\mathbb{N}^{+} \times\{0, 1\}$. Similarly, one can easily show that $Q(\delta,u;\tau)$ is submodular on $\mathbb{N}^{+}\times\{1, 2\}$, and $Q(\tau,u;\delta)$ is submodular on $\mathbb{N}^{+}\times\{0, 1\}$. Since the objective is to minimize the average MSE, renewing the channel before transmission provides no estimation benefit. Therefore, for any fixed AoI, the optimal policy is weakly increasing in AoC (from idle to transmit to renewal). This completes the proof.

\balance
\section*{References}
\vspace{-2 em}
\bibliographystyle{IEEEtran}
\bibliography{ref}

\end{document}